\documentclass{llncs}
\usepackage{amsmath}
\usepackage{amssymb}
\usepackage{multirow}
\usepackage{tabularx}

\begin{document}

\title{An asymmetric primitive based on the Bivariate Function Hard Problem}

\author{Muhammad Rezal Kamel Ariffin \inst{1,}\inst{2}}
\institute{Al-Kindi Cryptography Research Laboratory, Institute for
Mathematical Research, Universiti Putra Malaysia (UPM) \and
Department of Mathematics, Faculty of Science, Universiti Putra
Malaysia (UPM), Selangor, Malaysia \newline \textrm{ }\newline
\email{rezal@putra.upm.edu.my}}

\maketitle

\begin{abstract}
The Bivariate Function Hard Problem (BFHP) has been in existence
implicitly in almost all number theoretic based cryptosystems. This
work defines the BFHP in a more general setting and produces an
efficient asymmetric cryptosystem. The cryptosystem has a complexity
order of $O(n^2)$ for both encryption and decryption.
\end{abstract}



\section{Introduction} In Section 2 of this work, we define the
Bivariate Function Hard Problem (BFHP) and illustrate its existence
within the RSA hard problem. We then proceed to produce an
asymmetric cryptosystem in Section 3 that thoroughly utilizes the
BFHP concept. In Section 4 we produce a table of comparison between
known asymmetric algorithms and the algorithm introduced in this
work. We conclude in Section 5.

\section{Bivariate Function Hard Problem}

The following proposition gives a proper an analytical description
of the ``Bivariate Function Hard Problem" (BFHP).

\begin{proposition}
Let $F(x_{1},x_{2},...,x_{n})$ be a multivariate one-way function
that maps $F:\mathbb{Z}^{n}\rightarrow
\mathbb{Z}^{+}_{(2^{n-1},2^{n}-1)}$. Let $F_{1}$ and $F_{2}$ be such
functions (either identical or non-identical) such that
$A_{1}=F_{1}(x_{1},x_{2},...,x_{n})$,
$A_{2}=F_{2}(y_{1},y_{2},...,y_{n})$ and gcd$(A_{1},A_{2})=1$. Let
$u,v \in \mathbb{Z}^{+}_{(2^{m-1},2^{m}-1)}$. Let
\begin{equation}
G(u,v)=A_{1}u+A_{2}v
\end{equation}
If at minimum $m-n-1=129$, it is infeasible to determine $(u,v)$
from $G(u,v)$. Furthermore, $(u,v)$ is unique for $G(u,v)$ with high
probability.
\end{proposition}

\begin{remark}
Before we proceed with the proof, we remark here that the
diophantine equation given by $G(u,v)$ is solved when the parameters
$(u,v)$ are found. That is, the BFHP is solved when the parameters
$(u,v)$ are found.
\end{remark}

\begin{proof}
We begin by proving that $(u,v)$ is unique for each $G(u,v)$ with
high probability. Assume there exists $u_{1}\neq u_{2}$ and
$v_{1}\neq v_{2}$ such that
\begin{equation}
A_{1}u_{1}+A_{2}v_{1} = A_{1}u_{2}+A_{2}v_{2}
\end{equation}
We will then have
$$
Y=v_{1}-v_{2}=\frac{A_{1}(u_{1}-u_{2})}{A_{2}}
$$
Since gcd$(A_{1},A_{2})=1$ and $A_{2}\approx 2^{n}$, then the
probability that $Y$ is an integer is $2^{-n}$.\newline Next we
proceed to prove that to solve the diophantine equation given by
$G(u,v)$ is infeasible to be solved. The general solution for
$G(u,v)$ is given by
\begin{equation}
u=u_{0}+A_{2}t
\end{equation}
\noindent and
\begin{equation}
v=v_{0}-A_{1}t
\end{equation}
\noindent for some integer t. To find $u$ within the stipulated
interval ($u\in (2^{m-1},2^{m}-1)$) we have to find the integer $t$
such that $2^{m-1}<u<2^{m}-1$. This gives
$$
\frac{2^{m-1}-u_{0}}{A_{2}}<t<\frac{2^{m}-1-u_{0}}{A_{2}}.
$$
Then the difference between the upper and the lower bound is
$$\frac{2^{m}-1-2^{m-1}}{A_{2}}=\frac{2^{m-1}-1}{A_{2}}\approx \frac{2^{m-2}}{2^n}=2^{m-n-2}.$$
Since $m-n-1=129$, then $m-n-2=128$. Hence the difference is very
large and finding the correct $t$ is infeasible. This is also the
same scenario for $v$.$\blacksquare$
\end{proof}

\begin{remark}
It has to be noted that the BFHP in the form we have described has
to be coupled with other mathematical considerations upon
$F_{1},F_{2},u,v$ to yield practical cryptographic constructions.
\end{remark}

\begin{definition}
Let the tuple $(M,e,d,p,q)$ be strong RSA parameters. Let $N=pq,
ed\equiv 1 (\textrm{mod }\phi(N))$ and $\phi(N)=(p-1)(q-1)$. From
$C\equiv M^{e}(\textrm{mod }N)$ we rewrite as
\begin{equation}
C(M,j)=M^{e}-Nj
\end{equation}
where $j$ is the number of times $M^e$ is reduced by $N$ until
$C(M,j)$ is obtained. The problem of determining $(M,j)$ from
equation $(5)$ is the RSA BFHP. The pair $(M,j)$ is unique with high
probability for each $C(M,j)$.
\end{definition}

\begin{remark}
With little effort, one can also produce a BFHP for the discrete log
problem (DLP). Analysis could also be done within the framework
given for the RSA-BFHP.
\end{remark}

\noindent The following 3 analytical results gives a clear picture
regarding the RSA-BFHP. All result re-affirms the ``infeasibility"
of trying solve the RSA problem. We also produce a corollary that
may shed some light regarding the RSA problem and integer
factorization.


\begin{lemma}
The RSA BFHP is infeasible to be solved.
\end{lemma}

\begin{proof}
Let $X=M^{e}$. From
\begin{equation}
C(X,j)=X-Nj
\end{equation}
the general solution is
$$
X=X_{0}-Nt
$$
\noindent and
$$
j=j_{0}+t
$$
for some $t\in \mathbb{Z}$. It is easy to deduce that the correct t
belongs in the interval $(2^{k(e-1)-1},2^{k(e-1)}-1)$. Current RSA
deployment has $k=1024$. Hence, to solve the RSA BFHP is
infeasible.$\blacksquare$
\end{proof}

\begin{lemma}
RSA problem $\equiv_{p}$ RSA BFHP
\end{lemma}

\begin{proof}
From $C\equiv M^{e}(\textrm{mod }N)$ if the RSA problem is solved
then $M$ is found. Hence, $j=\frac{M^{e}-C}{N}$ is also found. Thus,
the RSA BFHP is solved.\newline

From $C(X,j)=X-Nj$, if the RSA BFHP is solved means that $(M,j)$ is
found. Thus, the RSA problem is solved.$\blacksquare$
\end{proof}

\begin{corollary}
Solving RSA BFHP does not imply successful factoring of $N=pq$.
\end{corollary}

\begin{proof}
From Remark 1, if RSA BFHP is solved then $(M,j)$ is found. That is,
$$
M=\sqrt[e]{C+Nj}
$$
\noindent and
$$
j=\frac{M^{e}-C}{N}.
$$
It is obvious that the factoring of $N$ was not
obtained.$\blacksquare$
\end{proof}

\section{A new asymmetric cryptosystem based on the BFHP}

\subsection{Common values}
This scheme is to facilitate secure communication asymmetrically
between 2 parties namely A (Along) and B (Busu). For both of them
there will 2 sets of public parameters determined pre-communication
and a common $n$-bit prime number. The party that initiates the
communication will utilize the set $G_{1}=(g_{1}, g_{2})$ while the
other party will utilize the set $G_{2}=(g_{3}, g_{4})$. These
public parameters are co-prime to each other and belong in the
interval $(2^{n-1}, 2^{n}-1)$. In fact both parties will have keys
generated by both sets for the eventuality of either initiating
communication or accepting incoming information. In this work we
assume Along is initiating while Busu is accepting secure
information.\newline \textrm{ }\newline

$\bullet$ \textbf{Key Generation by Along -sender}\\
\newline
\indent INPUT: The public prime number $p$, the public sets $G_{1}$ and $G_{2}$.\\
\indent OUTPUT: A public key for sending information $e_{A}$, an
ephemeral private key $d_{A}$ for generating $e_{A}$ and a
secret pair $(\alpha_{1},\alpha_{2})$. \\
\begin{enumerate}
    \item Generate a random private key $d_{A}$ within the interval $(2^{n-1}, 2^{n}-1)$.
    \item Compute $\tilde{\alpha_{1}}\equiv g_{1}d_{A} (\textrm{mod }p)$.
    \item Compute $\tilde{\alpha_{2}}\equiv g_{2}d_{A} (\textrm{mod }p)$.
    \item Generate two random and distinct $n$-bit ephemeral keys
    $k_{A1}$ and $k_{A2}$.
    \item Compute the secret integers
    $\alpha_{1}=\tilde{\alpha_{1}}+k_{A1}p$ and
    $\alpha_{2}=\tilde{\alpha_{2}}+k_{A2}p$. Both $\alpha_{1}$ and
    $\alpha_{2}$ belong in the interval $(2^{2n-1}, 2^{2n}-1)$.
    \item Let $e_{A}=g_{3}\alpha_{1}+g_{4}\alpha_{2}$.
\end{enumerate}

$\bullet$ \textbf{Key Generation by Busu -recipient}\\
\newline
\indent INPUT: The public prime number $p$, the public sets $G_{1}$ and $G_{2}$.\\
\indent OUTPUT: A public key for receiving information $e_{B}$, an
ephemeral private key $d_{B}$ for generating $e_{B}$ and a
secret pair $(\beta_{1},\beta_{2})$. \\
\begin{enumerate}
    \item Generate a random private key $d_{B}$ within the interval $(2^{n-1}, 2^{n}-1)$.
    \item Compute $\tilde{\beta_{1}}\equiv g_{3}d_{B} (\textrm{mod }p)$.
    \item Compute $\tilde{\beta_{2}}\equiv g_{4}d_{B} (\textrm{mod }p)$.
    \item Generate two random and distinct $n$-bit ephemeral keys
    $k_{B1}$ and $k_{B2}$.
    \item Compute the secret integers
    $\beta_{1}=\tilde{\beta_{1}}+k_{B1}p$ and
    $\beta_{2}=\tilde{\beta_{2}}+k_{B2}p$. Both $\beta_{1}$ and
    $\beta_{2}$ belong in the interval $(2^{2n-1}, 2^{2n}-1)$.
    \item Let $e_{B}=g_{1}\beta_{1}+g_{2}\beta_{2}$.
\end{enumerate}

$\bullet$ \textbf{Encryption by Along}\\

\indent INPUT: The public key tuple $(e_{A}, e_{B})$, and the
message \textbf{M} which is $n$-bits long and less than $p$.
\newline
\textrm{ }\newline \noindent OUTPUT: The ciphertext $C$.
\begin{enumerate}
    \item Upon informing Busu of the intention to send secure data,
    Along receives Busu's public key $e_{B}$.
    \item Along then generates $e_{AB}\equiv d_{A}e_{B}\equiv
    (\textrm{mod }p)$.
    \item Along then generates the ciphertext
    $C_{1}=(M+e_{AB})(\textrm{mod }p)$.
    \item Next, Along produces $\textit{sk}=H(e_{AB})$ where $H$ is a collision resistant hash function.
    \item Along will then utilize a symmetric algorithm Enc, to produce
    $C_{2}=Enc_{\textit{sk}}(M)$.
    \item Along will relay $(C_{1}, C_{2}, e_{A})$ to Busu.
\end{enumerate}


$\bullet$ \textbf{Decryption by Busu}\\

\indent INPUT: The private key $d_{B}$ and the tuple $(C_{1}, C_{2},
e_{A})$.
\newline
\textrm{ }\newline \noindent OUTPUT: The message \textbf{M}.
\begin{enumerate}
    \item Upon receiving ciphertext Busu computes $e_{BA}\equiv d_{B}e_{A}(\textrm{mod
    }p)$.
    \item Busu then computes \textbf{M'}=$(C_{1}-e_{BA})(\textrm{mod }p)$.
    \item Busu then produces $\textit{sk}=H(e_{BA})$
    \item Busu then decrypts $C_{2}$ with its corresponding
    symmetric decryption algorithm Dec to produce
    $M=Dec_{\textit{sk}}(C_{2})$.
    \item If \textbf{M'}$\neq M$ then abort.
    \item Else output \textbf{M'} which is the message.
\end{enumerate}

\begin{proposition}
From the above mentioned algorithm $e_{AB}=e_{BA}$.
\end{proposition}

\begin{proof}
$e_{AB}\equiv d_{A}e_{B}\equiv
    (\tilde{\alpha_{1}}\tilde{\beta_{1}}+\tilde{\alpha_{2}}\tilde{\beta_{2}})\equiv
    (\tilde{\beta_{1}}\tilde{\alpha_{1}}+\tilde{\beta_{2}}\tilde{\alpha_{2}})\equiv
    d_{B}e_{A}(\textrm{mod }p)=e_{BA}$.$\blacksquare$
\end{proof}

\begin{proposition}
The encryption process as mentioned above is IND-CCA2 secure.
\end{proposition}

\begin{proof}
This is a sketch. Any change to $C_{1}$ would result in the
decrypted value from $C_{1}$ which is $M'$ which would differ from
$M=Dec_{\textit{sk}}(C_{2})$ with high probability. Hence,
abort.$\blacksquare$
\end{proof}

\begin{lemma}
The problem of determining the secret parameters either from $e_{A}$
or $e_{B}$ is a BFHP.
\end{lemma}

\begin{proof}
It is obvious that with high probability ($\approx 2^{-n}$ - since
each $g_{i}$ are co-prime to each other) that the secret parameters
in either $e_{A}$ or $e_{B}$ are unique. Also, the difference
between the secret and public parameters are $n$-bits, one can set
$n=128$.$\blacksquare$
\end{proof}

\section{Table of Comparison}
Let $|E|$ denote public key size. Let $|M|$ denote the message size.
For RSA and ECC we utilize its parameters within its IND-CCA2
design. In determining the ciphertext size $|C|$ we also included
the public keys to be transmitted (where applicable). Complexity
time is taken in base case scenario deployed via the Fast Fourier
Transform (FFT).
\begin{center}
\begin{tabular}{|c|c|c|c|c|c|}
  \hline
  Algorithm & Encryption & Decryption & Ratio  & Ratio & Remark\\
   & Speed & Speed & $|M|:|C|$ & $|M|:|E|$ & \\
  \hline
  RSA & $O(n ^{2}\textrm{ }log\textrm{ }n)$ & $O(n ^{2}\textrm{ }log\textrm{ }n)$ & $1:2$ & $1:2$ & 2 parameter
  ciphertext\\
   & & & & & of $n$-bits each\\
  \hline
  ECC & $O(n ^{2}\textrm{ }log\textrm{ }n)$ & $O(n ^{2}\textrm{ }log\textrm{ }n)$ & $1:3$ & $1:2$ & 2 parameter ciphertext of $n$-bits each \\
  & & & & & + 1 $n$-bit public key\\
  \hline
  NTRU & $O(n\textrm{ }log\textrm{ }n)$ & $O(n\textrm{ }log\textrm{ }n)$ & Varies \cite{hoffstein2} & N/A & \\
  \hline
  $\textrm{This work}$ & $O(n\textrm{ }log\textrm{ }n)$ & $O(n\textrm{ }log\textrm{ }n)$ & $1:5$ & $1:3$  & 2 parameter ciphertext of $n$-bits \\
   & & & & & each + 1 $3n$-bit public key\\
  \hline
\end{tabular}
\end{center}

\begin{center}
{Table 1. Comparison table for input block of length $n$}
\end{center}

\section{Conclusion}
We conclude this work by stating that an efficient asymmetric
algorithm has been disclosed. By having complexity order of
$O(n\textrm{ }log\textrm{ }n)$ for both encryption and decryption,
it would cut $\approx \frac{2}{3}$ of the running time of algorithms
that do not achieve this speed. Furthermore, it achieves IND-CCA2
security.

\end{document}